\documentclass[9pt,a4paper,web]{article}
\usepackage[noadjust]{cite}
\usepackage{amsmath,amssymb,amsfonts}
\usepackage{algorithmic}
\usepackage{graphicx}
\usepackage[normalem]{ulem}
\usepackage{textcomp}

\usepackage{enumitem}
\usepackage{booktabs}
\newtheorem{thm}{Theorem}
\newtheorem{cor}{Corollary}
\newtheorem{definition}{Definition}
\newtheorem{lem}{Lemma}
\newtheorem{asmn}{Assumption}
\newtheorem{prpty}{Property}

\newtheorem{defi}{Definition}
\newtheorem{props}{Proposition}

\newtheorem{exx}{Example}
\newtheorem{remm}{Remark}
\newtheorem{prop}{Property}
\newtheorem{proof}{Property}

\renewenvironment{proof}{\smallskip \noindent{\textbf{Proof:}}}{\hfill \hspace*{1pt}\hfill $\blacksquare$ \newline \noindent}
\newenvironment{proofname}[1]{\smallskip \noindent {\textbf{#1}}}{\hfill \hspace*{1pt}\hfill $\blacksquare$ \newline \noindent}
\newenvironment{assumption}{\begin{asmn}\rm }{\hfill \hspace*{1pt} \hfill $\square$\end{asmn}}
\newenvironment{assumption2}{\begin{asmn}\rm }{\hfill \hspace*{1pt}\end{asmn}}
\newenvironment{remark}{\begin{remm}\rm }{\hfill \hspace*{1pt} \hfill $\square$\end{remm}}
\newenvironment{example}{\begin{exx}\rm }{\hfill \hspace*{1pt} \hfill $\square$\end{exx}}

\newenvironment{proposition}{\begin{props}\rm }{\hfill \hspace*{1pt} \hfill $\square$\end{props}}
\newenvironment{lemma}{\begin{lem}\rm }{\hfill \hspace*{1pt} \hfill $\square$\end{lem}}

\newenvironment{theorem}{\begin{thm}\rm }{\hfill \hspace*{1pt} \hfill $\square$\end{thm}}

\newenvironment{definition2}{\begin{defi}\rm }{\hfill \hspace*{1pt} \end{defi}}

\newcommand{\dom}{\ensuremath{\mathrm{dom}}\,}
\newcommand{\diag}{\ensuremath{\mathrm{diag}}}

\newcommand{\K}{\ensuremath{\mathcal{K}}}
\newcommand{\KL}{\ensuremath{\mathcal{KL}}}

\newcommand{\real}{\ensuremath{{\mathbb R}}}
\newcommand{\complex}{\ensuremath{{\mathbb C}}}
\newcommand{\integer}{\ensuremath{{\mathbb Z}}}

\newcommand{\sign}{\ensuremath{\mathrm{sign}}}

\graphicspath{{figures/}}

\begin{document}
\title{Finite-time stability properties of Lur'e systems with piecewise continuous nonlinearities\thanks{Work supported by the ANR under grant HANDY  ANR-18-CE40-0010.}}
\author{S. Mariano\thanks{S. Mariano and R. Postoyan are with the Universit\'e de
	Lorraine, CNRS, CRAN, F-54000 Nancy, France (e-mails: firstname.name@univ-lorraine.fr).}, R. Postoyan\footnotemark[2], L. Zaccarian\thanks{L. Zaccarian is with LAAS-CNRS, Universit\'e de Toulouse, CNRS, Toulouse, France
	and the Department of Industrial Engineering, University of Trento, Trento, Italy (e-mail: zaccarian@laas.fr).}
	}
\maketitle

\begin{abstract}
We analyze the stability properties of Lur'e systems with piecewise continuous nonlinearities by exploiting the notion of set-valued Lie derivative for Lur'e-Postnikov Lyapunov functions. We first extend an existing result of the literature to establish the global asymptotic stability of the origin under a more general sector condition. We then present the main results of this work, namely additional conditions under which output and state finite-time stability properties also hold for the considered class of systems. We highlight the relevance of these results by certifying the stability properties of two engineering systems of known interest:  mechanical systems affected by friction and cellular neural networks.
\end{abstract}
 

\section{Introduction}
\label{sec:introduction-introduction}

Defining conditions to ensure stability properties of continuous-time linear systems subject to a cone-bounded nonlinear output feedback, namely, the so-called the Lur'e problem, has been widely investigated in the literature see, e.g., \cite{LurePostnikov, popov1961absolute, NonlinearKhalil,yakubovich2004stability}. 
This class of systems is ubiquitously used in various engineering domains, such as mechanical engineering to describe dynamical systems affected by friction and/or unilateral constraints \cite{Eindhoven17}, electrical and electronic engineering to capture the behavior of electrical circuits with switches or electronic devices \cite{acary2008numerical,vasca2009new}, or neural networks \cite{soykens1999lur}; see \cite{brogliato2020dynamical} for additional examples. However, to the authors' best knowledge, very few results are available on the finite-time stability properties of Lur'e systems, see \cite{tang2017finite}, which concentrates on cluster synchronization of networks of Lur'e systems. 
Finite-time stability properties are gaining increasing attention due to their relevance in many applications such as high-order sliding mode algorithms \cite{polyakov2015finite}, controllers for mechanical systems \cite{bartolini2003survey}, spacecraft stabilization \cite{vorotnikov2002partial}, observer design problems \cite{andrieu2019lmi}; see \cite{vorotnikov1998partial} for additional examples. There is therefore a need for analytical tools to establish finite-time stability properties for this class of systems. In this context,  we investigate the output and state finite-time stability properties of Lur'e system with piecewise continuous nonlinearities.

Historically,  two different types of Lyapunov functions have been used to analyze the (absolute) stability of continuous-time Lur'e systems: quadratic functions of the state and the so-called Lur'e-Postnikov Lyapunov functions, which are the sum of a quadratic function of the state and a weighted sum of the integrals of the feedback nonlinearities \cite{NonlinearKhalil}. Lur'e-Postnikov Lyapunov functions are generally used to draw less conservative sufficient stability conditions \cite{yakubovich2004stability}. However, when the nonlinearities are piecewise continuous, as in, e.g., mechanical systems \cite{Eindhoven17}, neural networks \cite{forti2003global}, see also \cite{brogliato2020dynamical}, the challenge is that Lur'e-Postnikov Lyapunov functions become only differentiable almost everywhere (being locally Lipschitz continuous) due to the discontinuity points of the nonlinearities. Indeed, when the system nonlinearities are piecewise continuous and a Lur'e-Postnikov Lyapunov function is considered,  the standard tools used in the nonsmooth analysis, like Clarke's generalized directional derivatives, may lead to conservative algebraic Lyapunov conditions as we show in this paper; see also \cite{DellaRossaPhD21}. This limitation is overcome in \cite{Eindhoven17}, where trajectory-based arguments are used to prove an input-to-state (ISS) stability property, but no finite-time stability property is provided.

In this work, we first extend one of the results in \cite{Eindhoven17} to establish the global asymptotic stability of the origin for Lur'e systems with piecewise continuous nonlinearities under a more general sector condition. We resort for this purpose to a nonsmooth Lur'e-Postnikov Lyapunov function. We present algebraic Lyapunov decrease conditions by using the notion of set-valued Lie derivative \cite{BacCer99,Valadier89}. The set-valued Lie derivative is the key to overcoming the conservatism which the customarily used Clarke's generalized directional derivative may give, as we illustrate in a dedicated example. It has to be noted that in 
\cite{tang2017finite} set-valued Lie derivatives are also used in the analysis of these interconnections, however, the Lyapunov function is quadratic (thus continuously differentiable), which, as mentioned above, leads to more conservative conditions and, more importantly, the problem setting is different. Our main results establish  output and state finite-time stability properties for the considered Lur'e systems. To illustrate the usefulness of our results we focus on two engineering applications, considered respectively in \cite{Eindhoven17,forti2003global} and that can be modeled as Lur’e systems. Indeed, we establish output finite-time and state-independent local asymptotic stability properties for mechanical systems subject to friction, which is a novelty compared to \cite{Eindhoven17}. Furthermore, we certify that the cellular neural networks modeled as in \cite{forti2003global} are state finite-time stable, thus retrieving the results in \cite[Thm. 4]{forti2003global} while coping with a more general class of Lur'e systems.

The rest of the paper is organized as follows. Notation and background material are given in Section~II. The class of Lur'e systems under consideration is introduced in Section~III. In Section~IV, we address asymptotic stability characterizations with a novel algebraic Lyapunov proof. Finite-time stability results are given in Section~V, while we discuss applications of these results in  Section~VI. In Section~VII we give conclusions and some perspectives.






\section{Notation} 
\label{sec:notation}
Let $\real$ be the set of real numbers, $\real_{\geq0}:=[0,\infty)$, $\real_{>0}:=(0,\infty)$,  $\integer_{\geq 0}:=\{0,1,\dots\}$, $\integer_{>0}:=\{1,2,\dots\}$ and $\complex:=\{a+\mathsf{i} b: a,b \in \real\}$ with $\mathsf{i}:=\sqrt{-1} $. The notation $\real^n$ stands for the $n$-dimensional Euclidean space with $n\in\integer_{>0}$. The notation $\mathbb{B}_n$ stands for the closed unit ball of $\real^n$ centered at the origin and we write $\mathbb{B}$ when its dimension is clear from the context. We denote with $\emptyset$ the empty set. Given a vector $x\in\real^n$, we denote with $x_{\ell}$ its $\ell$-th element, $\ell \in \{1,\dots, n\}$, and with $|x|$ its Euclidean norm. The notation $\boldsymbol{0}_n$ stands for the vector of $\real^n$,  whose $n\in \integer_{>0}$ elements are all equal to $0$. We use $I_n$ to denote the identity matrix of dimension $n \times n$ with $n\in\integer_{>0}$ while  $O_n$ denotes the null matrix of dimension $n \times n$ with $n\in\integer_{>0}$. Given two vectors $x_1\in\real^n$ and $x_2\in\real^m$ with $n,m\in\integer_{>0}$, we denote $(x_1,x_2):= [x_1^\top x_2^\top]^\top$ for the sake of convenience. Given a matrix $A\in\real^{n \times m}$ with $n,m\in\integer_{>0}$, $A_{\ell}$ stands for its $\ell$-th row where $\ell \in \{1,\dots, n\}$, $|A|$ is its spectral norm while $\ker(A)$ stands for its kernel. The notation $\diag(x_1, \dots, x_n)$ stands for the diagonal matrix of $\real^{n\times n}$ whose $n\in \integer_{>0}$ diagonal elements are $x_1, \dots, x_n\in \real$. We define a symmetric matrix $P\in\real^{n \times n}$ with $n\in\integer_{>0}$ to be positive (negative) definite, i.e., $P>0 \, (P<0)$, if all its eigenvalues are real and positive (negative); we say that $P$ is positive (negative) semidefinite, i.e., $P\geq0 \, (P\leq0)$, if all its eigenvalues are real and  non-negative (non-positive). Given a set $\mathcal{S} \subset \real^n$ with $n\in\integer_{>0}$, $\overline{\text{co}} \, \mathcal{S}$ is its closed convex hull. Given a function $f: X \rightarrow Y$, the domain of $f$ is defined as $\dom f = \{x \in X : f(x)\neq \emptyset\}$. A function $f: \mathcal{X} \rightarrow \real_{\geq 0}$ with $\mathcal{X} \subseteq \real^n $ and $n\in\integer_{>0}$ is radially unbounded  if $f(x) \rightarrow \infty$  as $|x| \rightarrow \infty$. Let $f: \real^n \rightarrow  \real$ and $r \in \real$ with $n\in\integer_{>0}$, we denote by $f(r)^{-1}$ the set $\{x \in \real^n : f(x) = r\}$, which may be empty.   Let $X$ and $Y$ be two non-empty sets,  $T:X \rightrightarrows Y$ denotes a set-valued map from $X$ to  $Y$. We will refer to class $\K$, $\K_\infty$ and $\KL$ functions as defined in \cite[Chap. 3]{TeelBook12}.  Let $f: \real \rightarrow  \real$ and $s_\circ \in \real$, then  $f'(s_\circ):=\lim_{s\rightarrow s_\circ}(f(s)-f(s_\circ))/(s-s_\circ)$, when it exists.
A function $f:\real \to \real$ is  \emph{piecewise continuous} if for any given interval $[a,b]$, with $a<b\in \real$, there exist a finite number of points $a\leq x_0<x_1<x_2<\dots<x_{k-1}<x_{k}\leq b$ with $k\in\integer_{\geq0}$  such that $f$ is continuous on $(x_{i-1},x_i)$ for any $i\in\{1,\dots,k\}$ and its one-sided limits exist as finite numbers. A function $f:\real \to \real$ is \emph{piecewise continuously differentiable} if $f$ is continuous and for any given interval $[a,b]$, with $a<b\in \real$, there exists a finite number of points $a\leq x_0<x_1<x_2<\dots<x_{k-1}<x_{k}\leq b$, with $k\in\integer_{\geq0}$  such that $f$ is continuously differentiable on $(x_{i-1},x_i)$ for any $i\in\{1,\dots,k\}$ and the one-sided limits $\lim_{s\rightarrow x_{i-1}^+} f'(s)$ and $\lim_{s\rightarrow x_{i}^-} f'(s)$ exists for any $i\in\{1,\dots,k\}$. 

\section{Problem statement}
Consider the system of the form
\begin{align}
    \nonumber
	\dot x &= A x + Bu \\
	\label{sys_non_reg}	
	 y&=Cx \\
	 \nonumber
     u&=-\boldsymbol\psi(y),	 
\end{align}
where $x\in \real^n$ is the state, $u, y\in \real^p$ are respectively the input and the output and $A$, $B$ and $C$ are real matrices of appropriate dimensions. The function $\boldsymbol\psi:\real^p \rightarrow \real^p$ is decentralized, namely for any $ y=(y_1,\dots,y_p)\in \real^p$, $\boldsymbol\psi(y)=(\psi_1(y_1), \dots, \psi_p(y_p))$. We suppose that $\psi$ satisfies the next sector condition.
\begin{assumption}
\label{ass:dec_nl}
For any  $i \in \{1, \dots, p\}$, $\psi_i$ is piecewise continuous and there exists $\zeta_i\in(0,+\infty]$ such that
    \begin{equation}
    \label{eq:sec_cond}
    \psi_i(y_i)(\psi_i(y_i)-\zeta_i y_i) \leq 0,    \qquad  \forall y_i \in \real.
    \end{equation}
\end{assumption}
The sector condition \eqref{eq:sec_cond} is more general than the one considered in \cite{Eindhoven17}, that is recovered when $\zeta_i=+\infty$ for all $i \in \{1, \dots, p\}$, in which case \eqref{eq:sec_cond} reads
\begin{equation}
    \label{eq:sec_cond_infty}
    -\psi_i(y_i) y_i \leq 0,    \qquad  \forall y_i \in \real \qquad \forall i \in \{1, \dots, p\}.
\end{equation}   
This generalization allows to derive less conservative stability conditions when the nonlinearities satisfy \eqref{eq:sec_cond} with some finite $\zeta_i$. Assumption~\ref{ass:dec_nl} characterizes a so-called Lur'e system \cite[Ch. 7]{NonlinearKhalil}, \cite{yakubovich2004stability}. 

In view of Assumption~\ref{ass:dec_nl}, system \eqref{sys_non_reg} may have a discontinuous right-hand side. Therefore, when we refer to the solutions to system \eqref{sys_non_reg}, we consider its so-called (generalized) Krasovskii solutions, which coincide with the solutions obtained by the Krasovskii regularization \cite{hajek1979discontinuous} of  \eqref{sys_non_reg}, that is    
\begin{align}
	\label{sys_reg_eq}
	\dot x &\in F(x):= A x - B \boldsymbol\Psi(Cx), \qquad y=Cx,		
\end{align}
where $\boldsymbol\Psi(y)=(\boldsymbol\Psi_1(y_1),\dots,\boldsymbol\Psi_p(y_p))$ is the Krasovskii regularization of  $\boldsymbol\psi$ in \eqref{sys_non_reg},  whose components are defined as $\boldsymbol\Psi_i(y_i):= \bigcap\limits_{s>0} \overline{\text{co}} \, \boldsymbol{\psi_i}(y_i + s \mathbb{B})$, $i\in \{1,\dots,p\}$, for any $y_i \in \real$. Observe that, by Assumption~\ref{ass:dec_nl}, $F$ is outer semicontinuous and locally bounded on $\real^n$ and $F(x)$ is convex for any $x \in \real^n$, thus local existence of solutions to \eqref{sys_reg_eq} is guaranteed by Theorem~3 in \cite[Ch. 2.1]{aubin2012differential}. Moreover, by definition, each $\boldsymbol\Psi_i:\real \rightrightarrows \real$ is set-valued only on a set of isolated points, therefore it is locally integrable: a property that will be exploited in the following.

We analyze the stability properties of system \eqref{sys_reg_eq} in the sequel, thereby ensuring the same stability properties for the Krasovskii solutions of \eqref{sys_non_reg}. As customary in the Lur'e systems literature and as shown in Fig.~\ref{fig:int_sys}, we perform a loop transformation to interpret  system \eqref{sys_reg_eq} as the feedback interconnection of two passive systems.

By following the steps in \cite[Ch. 7.1.2]{NonlinearKhalil} and \cite{Eindhoven17} and by adopting the same mathematical notation found in \cite[Ch. 7]{NonlinearKhalil}, we define the dynamic multiplier with transfer function
\begin{equation}
	\label{eq:dyn_multi_tf_reg}
	\mathcal{M}(s):=I + \Gamma s \quad  \forall s\in \complex,
\end{equation} 
where $\Gamma:=\diag(\gamma_1, \dots, \gamma_p)$ and  $\gamma_1, \dots, \gamma_p >0$ are suitable parameters, as detailed in the sequel. We thus interpret system \eqref{sys_reg_eq} as the feedback interconnection of the linear system $\Sigma_1$
\begin{equation}
    \label{sys_lin_reg}
    \Sigma_1 :\,\,
    \begin{cases}
	\, \dot x \!\!\!\! &= A x + B u \\
	\, \overline y \!\!\!\! & =  (C + \Gamma C A) x + (\Gamma C B + Z) u, 	
	\end{cases}
\end{equation}
where $Z:=\diag(\zeta_1^{-1}, \dots, \zeta_p^{-1})$ with\footnote{When $\zeta_i=+\infty$ we use the convention $\zeta_i ^{-1}=0$.} $\zeta_i \in (0, +\infty]$ in Assumption~\ref{ass:dec_nl}, with the nonlinear system $\Sigma_2$ \begin{equation}
    \label{sys_nonlin_reg}
	\Sigma_2 :\,\,
	\begin{cases}
	\, \dot y \!\!\!\! &= g(y,\overline y):= -\Gamma^{-1} y+ \Gamma^{-1}(\overline y -  Z u)  \\	
	\, u  \!\!\!\! &\in  - \boldsymbol\Psi(y). 	
	\end{cases}
\end{equation}
We assume that system $\Sigma_1$ is strictly passive from $u$ to $\overline y$ with quadratic storage function $ U_1$ defined as 
\begin{equation}
    \label{eq:storage1}
    U_1(x):=\frac{1}{2} x^\top P x, \quad \forall x \in \real^n,
\end{equation}
with $P \in \real^{n\times n}$ symmetric and positive definite, as formalized next.
\begin{assumption}
    \label{ass:strict_pass}
     System $\Sigma_1$ in \eqref{sys_lin_reg} is strictly passive from $u$ to $\overline y$ with storage function $U_1$ in \eqref{eq:storage1} \cite[Def. 6.3]{NonlinearKhalil}, i.e., there exist matrices $\Gamma>0$ diagonal, $P=P^\top>0$ and a scalar $\eta>0$ such that 
	\begin{equation}
	\label{eq:diss_sys1}
	M:=\begin{bmatrix}
	P A + A^\top P + \eta I_n &
	P B - (C + \Gamma C A)^\top \\
	B^\top P -  (C + \Gamma C A)&
	-2Z - \Gamma C B - (\Gamma C B)^\top 
	\end{bmatrix} \leq 0.
	\end{equation}
\end{assumption}
The linear matrix inequality \eqref{eq:diss_sys1} in Assumption~\ref{ass:strict_pass} can be efficiently tested numerically. Several tools are also available in the literature to certify \eqref{eq:diss_sys1}: the Kalman-Yakubovich-Popov lemma \cite[Lemma 6.3]{NonlinearKhalil} and the equivalent conditions given in \cite[Ch. 3.1]{brogliato2007dissipative} for minimal realizations or the results surveyed in \cite[Ch. 3.3]{brogliato2007dissipative} for nonminimal ones, to cite a few.

On the other hand, it can be proven, as clarified later in Remark~\ref{rem:sys2_passive} in Section~\ref{subsec:set_valued Lie}, that $\Sigma_2$ in \eqref{sys_nonlin_reg} is passive from input $\overline y$ to output $-u$, by considering the piecewise continuously differentiable storage function $U_2$ defined as
\begin{equation}
    \label{eq:storage2}
    U_2(y):=\sum_{i=1}^{p} \gamma_i \int_{0}^{y_i} \psi_i(\sigma) d\sigma, \quad \forall y \in \real^p,
\end{equation}
 where $\gamma_i > 0$ for $i\in \{1,\dots,p\}$ are the diagonal elements of $\Gamma$ as defined after \eqref{eq:dyn_multi_tf_reg}. 
\begin{figure}
    \centering
	\includegraphics[width=0.40\textwidth]{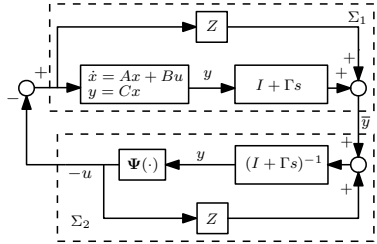} 
	\caption{System \eqref{sys_reg_eq} as the feedback interconnection of  systems $\Sigma_1$ and $\Sigma_2$.}
	\label{fig:int_sys}
\end{figure} 

We are now ready to proceed with the stability analysis of~\eqref{sys_reg_eq}. First, we provide sufficient conditions to ensure global asymptotic stability of the origin for system \eqref{sys_reg_eq} in Section~\ref{sec:asy_stability}. We then analyze its finite-time stability properties in Section~\ref{sec:ftp}.

\section{Asymptotic stability}
\label{sec:asy_stability}
\subsection{Nonsmooth Lur'e-Postnikov Lyapunov functions}
Inspired by \cite{NonlinearKhalil, Eindhoven17} where Lur'e systems with continuous nonlinearities are considered, we characterize the stability of the origin for system \eqref{sys_reg_eq} with a Lur'e-Postnikov Lyapunov function $V$ given by
\begin{align}
		\nonumber
		V(x)&:= U_1(x) + U_2(Cx)\\\label{Lyap_V_original_nl}
		&=\frac{1}{2} x^\top P x +  \sum_{i=1}^{p} \gamma_i \int_{0}^{C_i x} \psi_i(\sigma) d\sigma, \quad \forall x\in \real^n,
\end{align}
where $P$ comes from Assumption~\ref{ass:strict_pass}. Function $V$ is piecewise continuously differentiable, and thus locally Lipschitz, therefore there are points where its gradient is not defined. A standard tool to circumvent this is Clarke's generalized directional derivative, defined for each direction $f\in \real^n$ at each $x\in \real^n$ as \cite[page 11]{Clarke90}
$$V^\circ (x;f):=\max \{\langle v,f\rangle : v \in \partial V(x) \},$$
where $\partial V(x)$  denotes Clarke's generalized gradient of $V$ at $x$ given by
\begin{equation}
\label{eq:ClarkeGradient}
\partial V(x):=\{v \in \real^n| v \in  x^\top P + ({\boldsymbol{\Psi}}(Cx))^\top \Gamma C \}.
\end{equation}
However, the Lyapunov analysis of system~\eqref{sys_reg_eq} using Clarke's generalized directional derivative of $V$ is often too conservative  to establish asymptotic stability of the origin. Roughly speaking, for some $x\in\real^n \setminus\{\boldsymbol{0}_n\}$ there may exist a selection $f_{\text{bad}}\in F(x)$  that is never viable for any solution to \eqref{sys_reg_eq} and such that $V^\circ(x,f_\text{bad})>0$, thereby preventing to prove that the origin of the system is globally asymptotically stable, as illustrated in the next example.
\begin{example}
\label{rem:Clarke_problem_pt1}
  Consider system \eqref{sys_reg_eq} with $n=2$,  $p=1$ (SISO case), 
	\begin{equation*}
	\label{eq:cex}
	A= \begin{bmatrix}
	-1 &
	-1\\
	1&
	-1 
	\end{bmatrix}, \quad
	B= \begin{bmatrix}
	1 \\
	0
	\end{bmatrix}, \quad C= \begin{bmatrix}
	1 &
	0
	\end{bmatrix},
	\end{equation*}
	and where $\boldsymbol\Psi$ is the Krasovskii regularization of $\psi:\real \rightarrow [-\tfrac{1}{4},1]$, defined as  $\psi(s)=1$ if  $s>0$, $\psi(s)=-\tfrac{1}{4}$ if  $s<0$, and $\psi(s)=0$ if $s=0$; hence $\boldsymbol\Psi(0)=[-1/4,1]$. This function $\psi$ satisfies Assumption~\ref{ass:dec_nl} with $\zeta_1=+\infty$ (namely \eqref{eq:sec_cond_infty}). Consider $V$ as in \eqref{Lyap_V_original_nl} with $P=I_2$ and $\Gamma=\gamma_1 = 1$.
  The proposed selection of matrices $A$, $B$, $C$, $P$ and $\Gamma$ and the set-valued map $\boldsymbol{\Psi}$ are such that Assumptions~\ref{ass:dec_nl} and \ref{ass:strict_pass} are satisfied.
 Function \eqref{Lyap_V_original_nl} in this case is given by  $V(x)=\frac{1}{2} (x_1^2 +x_2^2)+ \int_{0}^{x_1} \psi(\sigma) d\sigma$ for any $x \in \real^2$. We have that  $V$ is positive definite and radially unbounded. Furthermore, $V$ is not differentiable at $\{0\}\times\real$. By following \cite[Ch. 4]{clarke2008nonsmooth} as summarized in \cite[Ch. 2.4.2]{DellaRossaPhD21}, to analyze the stability of the origin for the considered system, we study at any $x\in \real^2$ the maximum of $V^\circ(x;f)$ over all allowable directions $f\in F(x)$ with $F$ as in \eqref{sys_reg_eq}. In this regard, consider $x=(0,\frac{1}{2})$,
  \begin{align}
  \nonumber
  \max\limits_{f \in F(0,\tfrac{1}{2})}  &V^\circ((0,\tfrac{1}{2});f)=\max \Big\{ \langle v,f \rangle| \\ 
  \label{eq:Vex}
  &v \in [-\tfrac{1}{4},1] \times \big\{\tfrac{1}{2}\big\}, f \in \big[-\tfrac{3}{2},-\tfrac{1}{4}\big] \times \big\{-\tfrac{1}{2}\big\}\Big\}=\tfrac{1}{8}.
  \end{align}
  With this positive upper bound, in view of \cite[Def. 2.16]{DellaRossaPhD21}, we cannot establish asymptotic stability of the origin\footnote{The Ryan's invariance principle \cite{ryan1998integral} is also not applicable to guarantee asymptotic stability of the origin for this example.} \cite[Thm. 2.18]{DellaRossaPhD21}. Nevertheless a direct inspection shows that $V$ strictly decreases along all solutions outside the origin. The issue is overcome in the following by exploiting the notion of set-valued Lie derivative of $V$ \cite{Valadier89}.
\end{example}

In \cite{Eindhoven17}, the authors overcame the limitations discussed in Example~\ref{rem:Clarke_problem_pt1} by using \emph{trajectory-based} Lyapunov arguments  when Assumption~\ref{ass:dec_nl} holds with $\zeta_i=+\infty$ for any $i\in \{1,\dots,p\}$. In the next theorem, we establish global asymptotic stability of the origin for system \eqref{sys_reg_eq}. Compared to \cite{Eindhoven17}, the result relies on the more general sector condition in \eqref{eq:sec_cond}, and, importantly for the sequel, its proof uses \emph{algebraic} Lyapunov arguments. 
\begin{theorem}
	\label{thm:KLbound}
	Consider system  \eqref{sys_reg_eq} and suppose that Assumptions~\ref{ass:dec_nl} and \ref{ass:strict_pass} hold. Then the origin is GAS, i.e., there exists $\beta \in \KL$ such that all solutions $x$ satisfy 
	\begin{equation}
	|x(t)|\leq \beta(|x(0)|,t), \quad \forall t\in \real_{\geq 0}.
	\label{eq:KL}
	\end{equation} 
\end{theorem} 
The proof of Theorem~\ref{thm:KLbound} is given in Section \ref{subsec:set_valued Lie}, where we use the concept of set-valued Lie derivative that we now recall.

\subsection{Set-valued Lie derivative and its properties}
\label{subsec:set_valued Lie}
The set-valued Lie derivative of $V$ with respect to $F$ in \eqref{sys_reg_eq} at $x \in \real^n$ is defined as \cite{BacCer99} 
\begin{align}
\dot {\overline V}_F(x) := \{ a\in \real | \; \exists f \in F(x):  \langle v,f \rangle = a,\; \forall v\in \partial V(x)\},
\label{eq:LieD}
\end{align}
with $\partial V(x)$ given in \eqref{eq:ClarkeGradient}. Note that $\dot {\overline V}_F(x)$ is a subset of $\{\langle v,f \rangle | v \in \partial V(x), f \in F(x)\}$ and that, by definition, at any $x$ where $V$ is differentiable, so that $\partial V(x)$ is a singleton, this reduces to the set of all standard directional derivatives of $V$ in any direction of $f \in F(x)$. Notice that $\dot {\overline V}_F(x)$ may be the empty set as illustrated later in Example~\ref{rem:Clarke_problem_pt2}. In the next lemma, a useful and intuitive upper bound of the set-valued Lie derivative of $V$ in \eqref{Lyap_V_original_nl} along dynamics \eqref{sys_reg_eq} is provided.
\begin{lemma}
	\label{lem:LieDSingleton_Lure}
Given function $V$ in \eqref{Lyap_V_original_nl} and $F$ in \eqref{sys_reg_eq}, 
	\begin{equation}
	\label{eq:supLieDnn}
		\sup\dot {\overline V}_F(x) \leq\sup\limits_{\substack{u \in - \boldsymbol \Psi(Cx)}}\big ((x^\top P - u^\top \Gamma C)(Ax + Bu) \big ), \quad \forall x\in \real^n,
	\end{equation} 
 where we use the convention $\sup \emptyset = - \infty$ fot the left-hand side when $\sup\dot {\overline V}_F(x)=\emptyset$.
\end{lemma}
\begin{proof}
For each $x\in \real^n$ and each element $f=Ax +Bu \in F(x)$ with $u \in - \boldsymbol\Psi(Cx)$, as in \eqref{sys_reg_eq}, denote $\rho(x,f):=Px + C^\top \Gamma  u$ and note that $\rho(x,f) \in  \partial V(x)$. Notice that $\rho$ and $f$ are defined by selecting the same $u \in - \boldsymbol\Psi(Cx)$. In view of Lemma~8 in \cite{KuraMulti21}, exploiting this selection we have that $\sup\dot {\overline V}_F(x) \leq\sup\limits_{\substack{u \in -\boldsymbol\Psi(Cx)}}(\rho(x,f)^\top f)$, thus concluding the proof.
\end{proof} \\
\indent Exploiting \eqref{eq:LieD} and Lemma \ref{lem:LieDSingleton_Lure}, we can establish the next algebraic Lyapunov conditions for system \eqref{sys_reg_eq}. 
\begin{proposition}
	\label{prop:prop_V}
	Consider $F$ in \eqref{sys_reg_eq} and suppose that Assumptions~\ref{ass:dec_nl} and \ref{ass:strict_pass} hold.
	Then there exist $\alpha_1,\alpha_2,\alpha_3 \in \K_\infty$ such that function $V$ in \eqref{Lyap_V_original_nl} satisfies
	\begin{align}
	\label{eq:sendwitch_squared_nl}
	&\alpha_1(|x|) \leq V(x) \leq \alpha_2(|x|), \quad \forall x\in \real^n,\\
	\label{eq:prop_V_flow}
	&\sup\dot{\overline{V}}_F(x) \leq  -  \alpha_3(V(x)), \quad \forall x\in \real^n.
	\end{align}   
\end{proposition}
\begin{proof}
  From \eqref{eq:sec_cond} and \eqref{Lyap_V_original_nl}, $V$ is positive definite, continuous on $\real^n$ and radially unbounded. Therefore, \eqref{eq:sendwitch_squared_nl} holds by \cite[Lemma 4.3]{NonlinearKhalil}. Let  $x\in \real^n$, we have from Lemma~\ref{lem:LieDSingleton_Lure} that
	\begin{align}
	\nonumber
	\sup\dot{\overline{V}}_F(x) \leq \sup\limits_{\substack{u \in - {\boldsymbol{\Psi}}(Cx)}} & \Big [ (x^\top P - u^\top \Gamma C)(Ax + Bu)\\ 
	\nonumber
	&-  u^\top(-Z u - C x) - \frac{\eta}{2} |x|^2 \\
	\label{eq:LieVmid1}
	&+  u^\top (-Z u - C x) +\frac{\eta}{2} |x|^2 \Big ],
	\end{align}
with $Z$ as in   \eqref{sys_lin_reg}. Therefore,
	\begin{align}
	\label{eq:Vdotfinpwc}
	\sup\dot{\overline{V}}_F(x) \leq   &\sup\limits_{\substack{u \in - {\boldsymbol{\Psi}}(Cx)}}
	\left(\frac{1}{2} \begin{bmatrix}	x\\
	u \end{bmatrix}^\top  
	M 
	\begin{bmatrix}	x\\
	u \end{bmatrix}   
	+ u^\top(Z u + C x) \right)- \frac{\eta}{2}|x|^2,
	\end{align}
	with $M$ as in \eqref{eq:diss_sys1}. 
In view of  Assumption~\ref{ass:dec_nl}, it  holds that  $u^\top(Z u + C x) \leq 0$ for all $u\in - \boldsymbol\Psi (Cx)$ and any $x \in \real^n$, as $\boldsymbol\Psi(Cx)$ is convex. Moreover, $M\leq 0$ in view of  Assumption~\ref{ass:strict_pass}. Therefore, we have from \eqref{eq:Vdotfinpwc}
	\begin{align}
    \nonumber
	\sup\dot{\overline{V}}_F(x) & \leq  \sup\limits_{\substack{u \in - {\boldsymbol{\Psi}}(Cx)}} ( u^\top(Z u + C x))- \frac{\eta}{2} |x|^2\\\label{eq:vfin}
	&\leq -\frac{\eta}{2} |x|^2 \leq -\frac{\eta}{2}(\alpha^{-1}_2(V(x)))^2 =: - \alpha_3(V(x)),
	\end{align}
	with $\alpha_3 \in \K_\infty$, which shows \eqref{eq:prop_V_flow} and the proof is complete.   
\end{proof}
\indent Based on Proposition~\ref{prop:prop_V} we provide an algebraic proof of Theorem~\ref{thm:KLbound}.

\begin{proofname}{Proof of Theorem~\ref{thm:KLbound}.}
Let  $x$ be a solution to \eqref{sys_reg_eq}. In view of \cite[Prop. 4]{KuraMulti21} and \cite[Lemma 2.20]{DellaRossaPhD21}, $V$ is non-pathological, and thus \cite[Lemma 2.23]{DellaRossaPhD21} ensures that $\frac{d}{d t}V(x(t)) \in  \dot {\overline V}_F(x(t))$ for almost all $t \in \dom x$. Hence, in view of  \eqref{eq:prop_V_flow} in Proposition~\ref{prop:prop_V}, we have that
	\begin{equation} 
	\label{eq:flow}
	\dot{V}(x(t)) \leq - \alpha_3(V(x(t))), \quad \text{for almost all \,}  t \in \dom x.
	\end{equation}
	By following the steps of the proof of \cite[Lemma A.4]{SontagIOS}, we have that $\dom x=\real_{\geq 0}$ and  there exists $\overline\beta \in \KL$ (independent of $x$) such that 
	\begin{equation}
	\label{eq:kl_v}
	    V(x(t))\leq \overline{\beta}(V(x(0)),t), \quad \forall t \in \real_{\geq 0}. 
	\end{equation}
	Equations \eqref{eq:sendwitch_squared_nl} and \eqref{eq:kl_v} imply
	    $|x(t)|  \leq \alpha_1 ^{-1}(V(x(t))) 
	     \leq \alpha_1 ^{-1}(\overline{\beta}(\alpha_2(|x(0)|),t))=:\beta(|x(0)|,t)$ for any $t \in \real_{\geq 0}$, with $\beta\in\KL$, thus concluding the proof.
\end{proofname}
\indent With the help of Theorem~\ref{thm:KLbound}, we can now establish that the origin of the system in Example~\ref{rem:Clarke_problem_pt1} is GAS.
\begin{example}
 \label{rem:Clarke_problem_pt2}
  The system in Example~\ref{rem:Clarke_problem_pt1} satisfies both Assumptions~\ref{ass:dec_nl} and \ref{ass:strict_pass} with the given selections of $Z$, $\Gamma$ and $P$. As a result $x=\boldsymbol{0}_2$ is GAS in view of Theorem~\ref{thm:KLbound}. It is instructive to see how the notion of set-valued Lie derivative helps overcoming the issue highlighted in Example~\ref{rem:Clarke_problem_pt1}. In particular, the set-valued Lie derivative of $V$ with respect to $F$ at $x=(0,\tfrac{1}{2})$ is the empty set. Indeed, for each $f \in F(0,\tfrac{1}{2})$ and any two different directions $v_1,v_2 \in \partial V (0,\tfrac{1}{2})$ with $v_1 \neq v_2$, we have $\langle f,v_1 \rangle \neq \langle f,v_2 \rangle$, thus there exists no $a \in \real$ satisfying the condition in \eqref{eq:LieD}. More specifically, given $F(0,\tfrac{1}{2})=\big[-\tfrac{3}{2},-\tfrac{1}{4}\big] \times \big\{-\tfrac{1}{2}\big\}$ and $\partial V(0,\tfrac{1}{2})=[-\tfrac{1}{4},1] \times \big\{\tfrac{1}{2}\big\}$, by selecting $v_1, v_2 \in \partial V(0,\tfrac{1}{2})$ with $v_1\neq v_2$, and $f \in F(0,\tfrac{1}{2})$  we have that $\langle f,v_1 \rangle=f_1 v_{1,1} -\tfrac{1}{4}$ and $\langle f,v_2 \rangle =f_1 v_{2,1} -\tfrac{1}{4}$ with $f_1 \in \big[-\tfrac{3}{2},-\tfrac{1}{4}\big]$ and $v_{1,1} \neq v_{2,1} \in [-\tfrac{1}{4},1]$. Therefore, $\langle f,v_1 \rangle=\langle f,v_2 \rangle$ if and only if $f_1(v_{1,1}-v_{2,1})=0$, which is impossible for the specified selection of $f$, $v_1$ and $v_2$. Hence, there exists no $a\in \real$  and $f \in F(0,\tfrac{1}{2})$ such that $\langle f,v \rangle=a$ for all $v \in \partial V (0,\tfrac{1}{2})$, thus implying that $\dot{\overline{V}}_F(0,\tfrac{1}{2})=\emptyset$. Besides this specific illustrative analysis, by exploiting Lemma~\ref{lem:LieDSingleton_Lure} we may actually show that $\sup\dot{\overline{V}}_F(x)<0$ for all $x\in \real^2 \setminus \{\boldsymbol 0_2 \}$. Indeed, we have that $\sup\dot{\overline{V}}_F(x)\leq \sup\limits_{\substack{u \in - {\boldsymbol \Psi}(x_1)}}(-x_1^2+2u x_1-x_2^2 +u x_2 -u^2)=\sup\limits_{\substack{u \in - {\boldsymbol \Psi}(x_1)}}\Big (-x_1^2+2u x_1 -\Big(\frac{1}{2}x_2 -u\Big)^2 -\frac{3}{4}x_2^2\Big)< 0$, because Assumption~\ref{ass:dec_nl} implies $u x_1 < 0$.
We, therefore, obtain that the supremum of the set-valued Lie derivative of $V$ with respect to $F(x)$ is strictly negative outside the origin, which was not possible to prove using the conservative upper bound \eqref{eq:Vex}.
\end{example}

We explain in the next remark why $\Sigma_2$ in \eqref{sys_nonlin_reg} is passive.
\begin{remark}
\label{rem:sys2_passive}
    System $\Sigma_2$ is passive from $\overline{y}$ to $-u$ as discussed at the end of Section~III. Let $g$ as in \eqref{sys_nonlin_reg}, we have that, by Lemma 8  in \cite{KuraMulti21} 
	\begin{equation}
	\label{eq:flow_passive}
\sup\dot{\overline{U}}_{2,g}(y)		\leq \sup\limits_{\substack{u \in - {\boldsymbol{\Psi}}(y)}}(u^\top (y+Zu) -  u^\top \overline  y), \quad \forall y\in \real^p.
\end{equation}
Let $y$ be a solution of \eqref{sys_lin_reg} with input $\overline y$ and output $u$. From \cite[Lemma 2.23]{DellaRossaPhD21}, $\frac{d}{dt}{U}_2(y(t)) \leq \sup\dot{\overline{U}}_{2,g}(y)$. Then, since $u^\top (Zu+y)\geq 0$ by Assumption \ref{ass:dec_nl} and \eqref{eq:flow_passive},  
\begin{equation}
\label{eq:diss_diff_sys2}
    \frac{d}{dt}{U}_2(y(t)) \leq \sup\limits_{\substack{u \in - {\boldsymbol{\Psi}}(y(t))}} -u^\top \overline y(t), \quad \text{for almost all } t\in \dom y.
\end{equation}
Hence, we obtain $U_2(y(t)) \leq  \, U_2(y(0)) + \int_{0}^{t} -u(s)^\top y(s)\,ds$ for all $t \geq 0$ with $t\in \dom y$ by integrating \eqref{eq:diss_diff_sys2}. Thus $\Sigma_2$ is passive from $\overline{y}$ to $-u$ as per Definitions~2.1 and 2.2 in \cite{sepulchre2012constructive}. Notice that, even if Definitions~2.1 and 2.2 are stated for single-valued outputs, we can apply these same definitions without loss of generality to our case where $u \in - \boldsymbol{\Psi}(y)$.
\end{remark}

We conclude this section with a discussion about how the conditions of Theorem~\ref{thm:KLbound} can be extended to.

\subsection{Extension under special properties of plant \eqref{sys_non_reg}}

Since the conditions in Theorem \ref{thm:KLbound} are only sufficient, we may prove that the origin is GAS for system \eqref{sys_reg_eq} via Lyapunov analysis by exploiting additional structural properties of matrices $A$, $B$ and $C$ in \eqref{sys_reg_eq}. A set of alternative exploitable properties for system \eqref{sys_reg_eq} is given next. 
\begin{prop}
\label{prop:alt_prop}
The following holds for system \eqref{sys_reg_eq}.
\begin{enumerate}[label=(\roman*)]
\item  Assumption \ref{ass:dec_nl} is satisfied.
\item  There exist matrices $\Gamma>0$ diagonal, $P=P^\top>0$ and a scalar $\eta>0$ such that
\begin{equation}
	\label{eq:diss_sys1n}
	\overline M:=\begin{bmatrix}
	P A + A^\top P + \eta I_n &
	P B \\
	B^\top P&
	-2Z - \Gamma C B - (\Gamma C B)^\top 
	\end{bmatrix} \leq 0.
\end{equation}
\item  There exist $H:=\diag(h_1, \dots, h_p)$  such that $\Gamma C A = HC$ and, for all $i \in \{1,\dots,p\}$, either $h_i \leq -1$ holds, or $h_i \leq 0$ and $Z=O_p$ holds, with $Z$ in \eqref{sys_lin_reg}. \hfill $\Box$
\end{enumerate}
\end{prop}
The conditions in items (ii) and (iii) in Property~\ref{prop:alt_prop} impose extra properties of the matrices $C$ and $A$ (item (ii)) and a different matrix inequality compared to \eqref{eq:diss_sys1} (item (iii)), indeed as the off-diagonal terms of $\overline M$ differ from those in $M$ in \eqref{eq:diss_sys1}. We show in the next lemma that Property~\ref{prop:alt_prop} implies GAS of the origin for system \eqref{sys_reg_eq}. We will invoke this extension in Section~\ref{sec:lure_app} to analyze the stability properties of the neural networks studied in \cite{forti2003global}.
\begin{lem}
\label{lem:alt_gas}
Suppose that system \eqref{sys_reg_eq} satisfies items (i)-(iii) of Property~\ref{prop:alt_prop}. Then the origin is GAS for system \eqref{sys_reg_eq}. \hfill $\Box$
\end{lem}
\begin{proof}
Let $x\in \real^n$ and consider $V$ in  \eqref{Lyap_V_original_nl}. We have from Lemma~\ref{lem:LieDSingleton_Lure} that, for any $x\in \real^n$, 
	\begin{align}
	\nonumber
	\sup\dot{\overline{V}}_F(x) \leq \sup\limits_{\substack{u \in - {\boldsymbol{\Psi}}(Cx)}} &\Big[(x^\top P - u^\top \Gamma C)(Ax + Bu)\\ 
	\label{eq:LieVmid1n}
	&-  u^\top Z u  - \frac{\eta}{2} |x|^2 	+  u^\top Z u  + \frac{\eta}{2} |x|^2 \Big].
	\end{align}
with $Z$ as in   \eqref{sys_lin_reg}. Therefore,
	\begin{align}
	\label{eq:Vdotfinpwcn}
	\sup\dot{\overline{V}}_F(x) \leq   &\sup\limits_{\substack{u \in - {\boldsymbol{\Psi}}(Cx)}}
	\left(\frac{1}{2} \begin{bmatrix}	x\\
	u \end{bmatrix}^\top  
	\overline M 
	\begin{bmatrix}	x\\
	u \end{bmatrix}   
	+ u^\top(Z u - H C x) \right )- \frac{\eta}{2}|x|^2.
	\end{align}
We note that, in view of  Assumption~\ref{ass:dec_nl}, it  holds that  $u^\top(Z u - HCx) \leq 0$ for all $u\in - \boldsymbol\Psi (Cx)$ and any $x \in \real$. Indeed, because each entry of $\boldsymbol\Psi (y)=\boldsymbol\Psi (Cx)$  is convex for any $x\in \real^n$, by \eqref{eq:sec_cond} it holds that $\alpha_i u_i y_i \leq - \frac{\alpha_i}{\zeta_i} u_i ^2$, and $ u_i y_i + \frac{1}{\zeta_i} u_i ^2 + \alpha_i u_i y_i - \alpha_i u_i y_i \leq  0$,  for any $\alpha_i \geq 0$, $u \in -\boldsymbol\Psi (Cx)$ and  $i\in\{1,\dots,p\}$. Therefore,  we have $ \frac{1}{\zeta_i} u_i ^2 + (1+\alpha_i) u_i y_i \leq  \alpha_i u_i y_i \leq - \frac{\alpha_i}{\zeta_i} u_i ^2 \leq0$. When $h_i \leq -1$, taking $\alpha_i=-1 - h_i \geq 0$, we deduce that, for any  $u \in -\boldsymbol\Psi (Cx)$ and  $i\in\{1,\dots,p\}$, $\frac{1}{\zeta_i} u_i ^2 - h_i u_i y_i \leq 0$ and thus $u^\top(Z u - HCx) \leq  0$ for all $u\in - \boldsymbol\Psi (Cx)$. In the particular case where $Z=O_p$, $-u^\top HCx \leq 0$ is true for any negative semidefinite matrix diagonal $H$ by \eqref{eq:sec_cond_infty}, for all  $u \in -\boldsymbol\Psi (Cx)$ and $i\in\{1,\dots,p\}$. Moreover, we assumed $\overline M\leq 0$ in item (ii) of Lemma~\ref{lem:alt_gas}. Therefore, similar to \eqref{eq:vfin}, from \eqref{eq:Vdotfinpwcn} we have 
	\begin{align}
    \label{eq:newvdiff}
	\sup\dot{\overline{V}}_F(x) & \leq -\frac{\eta}{2} |x|^2 \leq -\frac{\eta}{2}(\alpha^{-1}_2(V(x)))^2 =:  - \overline \alpha_3(V(x)),
	\end{align}
	with $\overline\alpha_3 \in \K_\infty$. Then, as anticipated, by exploiting \eqref{eq:sendwitch_squared_nl} and \eqref{eq:newvdiff}, and following similar steps of those in the proof of Theorem~\ref{thm:KLbound}, we conclude that the origin is GAS for system \eqref{sys_reg_eq}.    
\end{proof}
\section{Finite-time stability}
\label{sec:ftp}
\subsection{Definitions and assumptions} 
 In this section, we provide conditions to guarantee output and state finite-time stability properties for system \eqref{sys_reg_eq}. In particular, we consider the next stability notions, see \cite{SontagLOS00,zimenko2021necessary}.
\begin{definition2}
\label{def:output_stability}
Consider system \eqref{sys_reg_eq}. If its solutions are all forward complete\footnote{A solution is forward complete if its domain is unbounded \cite{angeli1999forward}.}, then we say that the system is:
\begin{enumerate}[label=(\roman*)]
    \item \label{def:oGAS}  \emph{output globally asymptotically stable} (oGAS) if there exists $\beta \in \KL$ such that for any solution $x$  
$$|y(t)| \leq \beta(|x(0)|, t), \qquad \forall t\in \real_{\geq 0};$$
    \item \label{def:SIoLAS}  \emph{state-independent output locally asymptotically stable} (SIoLAS) if there exist $r>0$  and $\beta \in \KL$ such that for all solution $x$,
	$$  |x(0)|<r \Rightarrow |y(t)| \leq \beta(|y(0)|, t), \qquad \forall t\in \real_{\geq 0};$$
	\item \label{def:OFTS}  \emph{output finite-time stable} (OFTS) if it is oGAS and for each solution $x$ there exists $0 \leq  T < +\infty$ such that $y(t)=\boldsymbol0_p$ for all $ t\geq T$;
	\item \label{def:SFTS}  \emph{state finite-time stable} (SFTS) if the origin is GAS and for each solution $x$ there exists $0 \leq  T < +\infty$ such that $x(t)=\boldsymbol0_n$ for all $ t\geq T$. \hfill $\square$
\end{enumerate}
\end{definition2}

To be able to prove the output stability properties in Definition~\ref{def:output_stability}, we make the next assumption. 
\begin{assumption2}                                     \label{ass:ofts}
	The following holds.
	\begin{enumerate}[label=(\roman*)]
	\item   Matrix  $CB$ is \emph{Lyapunov diagonally stable}  (LDS) \cite[Def. 5.3]{hershkowitz1992recent}, i.e., there exists a diagonal matrix $\overline \Gamma>0$ of appropriate dimensions such that $\overline \Gamma CB+(CB)^\top \overline \Gamma>0$.
	
	\item The origin is GAS for system \eqref{sys_reg_eq}.
	
	\item Each $\psi_i$, with $i\in \{1, \dots, p\}$, is discontinuous at the origin and both its left and right limits are non-zero, i.e., for any $i\in \{1, \dots, p\}$ $\lim\limits_{s \rightarrow 0^+} \psi_i (s) >0$  and $\lim\limits_{s \rightarrow 0^-} \psi_i (s) <0$. \hfill $\square$
	\end{enumerate}
\end{assumption2}
Item~(i) of Assumption~\ref{ass:ofts} imposes extra conditions on the matrices $C$ and $B$ of system~\eqref{sys_non_reg}. Sufficient conditions to ensure item~(ii) of Assumption~\ref{ass:ofts} are provided in Theorem~\ref{thm:KLbound} and Lemma~\ref{lem:alt_gas}. Finally, item (iii) of  Assumption~\ref{ass:ofts} requires each $\psi_i, i \in \{1, \dots, p\}$, to be non-zero at the origin and to have non-zero left and right limit at zero as well. Examples of engineering systems satisfying Assumption~\ref{ass:ofts} (as well as Assumptions~\ref{ass:dec_nl} and \ref{ass:strict_pass}) are provided in Section \ref{sec:lure_app}.
\subsection{Output and state finite-time stability}
We are now ready to present the main result of this section,  whose proof is given in Section~\ref{subsec:Proof_Thm2}. 
\begin{theorem}
\label{thm:OFTS}
	Consider system  \eqref{sys_reg_eq} and suppose that Assumptions \ref{ass:dec_nl} and \ref{ass:ofts} hold, then system \eqref{sys_reg_eq} is OFTS and SIoLAS.     
\end{theorem}

Theorem~\ref{thm:OFTS} establishes output finite-time stability properties for system~\eqref{sys_reg_eq}. A natural question is then whether \emph{state} finite-time stability properties can also be guaranteed. An answer to this question is given in the next theorem which establishes that, whenever Assumptions~\ref{ass:dec_nl} and \ref{ass:ofts} are satisfied, system~\eqref{sys_reg_eq} is SFTS if and only if $C$ is invertible.  
\begin{theorem}
	\label{thm:FTS}
	Consider  system  \eqref{sys_reg_eq} and suppose that Assumptions \ref{ass:dec_nl} and \ref{ass:ofts} are verified. Then the system is SFTS if and only if matrix $C$ is inveritble. 
\end{theorem}
\begin{proof}
We start  by proving that there exists $\varepsilon>0$ such that, for any $\xi\in \ker(C) \cap \varepsilon \mathbb{B}_n$, $u = - (CB)^{-1} CA \xi$  belongs to $\boldsymbol\Psi(\boldsymbol 0_p)$ and $CA\xi + CBu=\boldsymbol 0_p$. First, note that $CB$ is invertible as it is LDS by item (ii) of Assumption~\ref{ass:ofts}. Hence, for any $\xi\in \ker(C) \cap \varepsilon \mathbb{B}_n$, $u = - (CB)^{-1} CA \xi$ is well-defined and $CA\xi + CBu=\boldsymbol 0_p$. Secondly, in view of item (iii) of Assumption~\ref{ass:ofts}  there exists $\psi_\circ \in \real_{>0}$ such that $[-\psi_\circ,\psi_\circ]^p \subseteq \boldsymbol \Psi (\boldsymbol 0_p)$. Therefore, there exists $\varepsilon>0$ such that, for any $ \xi \in \ker(C) \cap \varepsilon \mathbb{B}_n$ and any $i \in \{1, \dots, p\}$, $|((CB)^{-1} CA)_i \xi| \leq \psi_\circ$, thus implying $u=-(CB)^{-1} CA \xi \in [-\psi_o,\psi_o]^p \subseteq \boldsymbol\Psi(\boldsymbol 0_p)$, as to be proven.

Now we are ready to prove the necessary and sufficient conditions of Theorem~\ref{thm:FTS}. The sufficient condition in Theorem~\ref{thm:FTS} is a direct consequence of Theorem~\ref{thm:OFTS}. We proceed by contradiction to prove the necessary condition in Theorem~\ref{thm:FTS}. We thus assume that $C$ is not invertible and consider $\varepsilon>0$ as at the beginning of this proof. Since for any $x\in \ker (C)\cap \varepsilon \mathbb{B}_n$ we can select $u = - (CB)^{-1} CA x$ that belongs to $\boldsymbol\Psi(\boldsymbol 0_p)$, we consider below solutions to \eqref{sys_reg_eq} satisfying
\begin{equation}
    \label{eq:particular_u}
    \dot x = Ax - B(CB)^{-1} CAx, \quad x\in \ker(C)\cap \varepsilon \mathbb{B}_n,
\end{equation} 
which implies
\begin{equation}
    \label{eq:particular_y}
    \dot y=C\dot x = (C A  - C B (CB)^{-1} CA) x=\boldsymbol 0_p, \quad x\in \ker(C) \cap \varepsilon \mathbb{B}_n.
\end{equation}
We now exploit \eqref{eq:particular_y} to attain a contradiction. By item (ii)  of Assumption~\ref{ass:ofts}, there exists $\delta>0$ such that any solution starting in $\delta \mathbb{B}_n$ does not leave $\varepsilon \mathbb{B}_n$ for all times. Let $x_p$ be a nonzero solution starting in $\ker(C) \cap \delta \mathbb{B}_n$, with output $y_p=C x_p$, which evolves according to \eqref{eq:particular_u} and \eqref{eq:particular_y}. Then  $y_p(0)=C x_p(0)=\boldsymbol 0_p$ and equation \eqref{eq:particular_y} imply  $y_p(t)=C x_p(t)=\boldsymbol 0_p$ and $\dot x_p (t) = (A - B(CB)^{-1} CA)x_p (t) \neq 0$ for all $t\geq 0$. As a consequence, $x_p$ exponentially converges to the origin but does not converge in finite-time. Such a solution establishes a contradiction, thus completing the proof.
\end{proof}
\indent We can now analyze the finite-time stability property of the system in Example~\ref{rem:Clarke_problem_pt1} in light of Theorems~\ref{thm:OFTS} and \ref{thm:FTS}.
\begin{example}
	 Consider the system in Example~\ref{rem:Clarke_problem_pt1}. Assumption \ref{ass:ofts} holds with $\overline \Gamma = 1$. As a result, the system is OFTS and SIoLAS. We also know from Theorem~\ref{thm:OFTS} that the system is not SFTS as $C$ is not invertible. Another way to see it is to consider $x(0)\in X:=\{0\} \times [-\frac{1}{4},\frac{1}{4}]$. A possible solution to \eqref{sys_reg_eq} is $x_p(t)=(0,x_2(0)e^{-t})$, which belongs to the set $X$ for all $t\geq 0$. Moreover, we have that $y_p(t)=0$ and $\dot y_p(t)=0$ for all $t\geq 0$. Clearly, $x_p$ converges exponentially to the origin, but not in finite-time.  
\end{example}
\subsection{Proof of Theorem~\ref{thm:OFTS}}
\label{subsec:Proof_Thm2}
The proof of Theorem~\ref{thm:OFTS} relies on the next lemma and proposition.
\begin{lemma}
	\label{lem:const_nl}
	Under Assumption~\ref{ass:dec_nl} and item~(iii) of Assumption~\ref{ass:ofts}, there exist $ \nu>0$ and $c>0$,  such that 
	\begin{equation}
	\label{eq:sec_cond_temp_2}
	 |u| \geq c,  \quad \forall u\in-\boldsymbol\Psi(y), \quad \forall y\in \nu \mathbb{B}_p \setminus \{\boldsymbol{0}_p\}.
	\end{equation}
\end{lemma}
\begin{proof}
	In view of item~(iii) of Assumption~\ref{ass:ofts}, there exist positive parameters $\nu_\circ$ and $c$ such that, for each  $i\in \{1, \dots, p\}$, $\psi_i$ is continuous in the intervals $[-\nu_\circ,0)$ and $(0,\nu_\circ]$, and $\min(|\lim\limits_{s \rightarrow 0^+} \psi_i (s)|,|\lim\limits_{s \rightarrow 0^-} \psi_i (s)|)\geq 2 c$. Hence, there exists $\nu \in (0,\nu_\circ]$ such that, for any $i\in \{1, \dots, p\}$ and $s\in[-\nu,0)\cup(0,\nu]$, $|\psi_i(s)|\geq  c$. Therefore, we have that for any $y\in \nu \mathbb{B}_p \setminus \{\boldsymbol{0}_p\}$ there exists $i\in\{1, \dots, p\}$ such that $|u|\geq|u_i|\geq c$ for all $u\in-\boldsymbol\Psi(y)$ thus concluding the proof.
\end{proof}

We also invoke the next proposition, which states algebraic properties of a piecewise continuously differentiable function, which is similar to the one in \eqref{eq:storage2}
\begin{equation}
\label{eq:aux_V}
W(Cx):=2 \sum_{i=1}^{p} \overline \gamma_i \int_{0}^{C_i x} \psi_i(\sigma) d\sigma, \quad \forall x\in \real^n,
\end{equation}
where $\overline  \gamma_1, \dots, \overline  \gamma_p>0$ are positve parameters selected such that $\overline  \Gamma CB+(CB)^\top  \overline  \Gamma>0$, with $\overline  \Gamma=\diag(\overline \gamma_1, \dots, \overline \gamma_p)$, which exist by item~(i) of Assumption~\ref{ass:ofts}. Function $W$ enjoys the following properties.
\begin{proposition}
\label{prop:aux_V_prop}
Suppose that Assumption \ref{ass:dec_nl} and items (i) and (iii) of Assumption \ref{ass:ofts} hold. Given function $W$ in \eqref{eq:aux_V}, there exist $\mu\in(0,\nu]$, with $\nu$ as in Lemma~\ref{lem:const_nl}, and $\alpha_4, \alpha_5 \in \K_\infty$ such that
\begin{align}
\label{eq:aux_V_sandwitch}
&\alpha_4(|Cx|)\leq W(Cx)\leq \alpha_5(|Cx|), &&\quad \forall x \in \mu \mathbb{B}_n, \\
\label{eq:dot_aux_V}
&\sup\dot{\overline{W}}_F(Cx)\leq -c \omega, &&\quad \forall x \in \mu  \mathbb{B}_n \setminus \ker(C),
\end{align} 
 with $c$ as in Lemma~\ref{lem:const_nl}, $\omega:=\lambda_1  (c -2\mu\frac{\lambda_2}{\lambda_1})>0$, $\lambda_1$ is the smallest eigenvalue of  $\overline \Gamma CB+(CB)^\top \overline \Gamma$, and $\lambda_2:=|\overline \Gamma C A|$.
\end{proposition}
\begin{proof}
From \eqref{eq:aux_V} and Lemma \ref{lem:const_nl}, for any $x \in \mu  \mathbb{B}_n \setminus \ker(C)$, $W(Cx)>0$ while  $W(Cx)=0$  for any $x \in \ker(C)\cap \mu  \mathbb{B}_n$. Moreover, we have that $W$ is continuous on $\mu \mathbb{B}_n$. Therefore, \eqref{eq:aux_V_sandwitch} holds in view of \cite[Lemma 4.3]{NonlinearKhalil}. Let $x \in \mu  \mathbb{B}_n \setminus \ker(C)$, from Lemma~\ref{lem:LieDSingleton_Lure}, by imposing $P=0$ and $\Gamma=\overline \Gamma$  in \eqref{eq:supLieDnn}, we have
	\begin{align}
		\label{eq:dot_w_1}
		\sup\dot{\overline{W}}_F(Cx) \leq \sup\limits_{\substack{u \in - {\boldsymbol{\Psi}}(Cx)}} (- 2 u^\top \overline \Gamma C(Ax + B u)).
	\end{align}  
	Using the Cauchy–Schwarz inequality, we obtain
	\begin{align}
	\nonumber
	\sup\dot{\overline{W}}_F(Cx) \leq &\sup\limits_{\substack{u \in - {\boldsymbol{\Psi}}(Cx)}} (- u^\top(  \overline  \Gamma CB+(CB)^\top  \overline  \Gamma)u \\
	\label{eq:dot_w_2}
	&+ 2 |\overline \Gamma C A||x||u| ).
	\end{align} 
	Thus,  in view of item~(i) of Assumption \ref{ass:ofts}, we have that
	\begin{align}
	\nonumber
	\sup\dot{\overline{W}}_F(x) &\leq \sup\limits_{\substack{u \in - {\boldsymbol{\Psi}}(Cx)}} (- \lambda_1  |u|^2 + 2|\overline \Gamma C A||x||u| ),\\
	\nonumber
	&=\sup\limits_{\substack{u \in - {\boldsymbol{\Psi}}(Cx)}} (-(\lambda_1|u| -2|\overline \Gamma C A||x|)|u|),\\
		\label{eq:dot_w_3}
	&\leq \sup\limits_{\substack{u \in - {\boldsymbol{\Psi}}(Cx)}} \Big (-\lambda_1 \big (|u| -2\frac{\lambda_2}{\lambda_1}|x|\big )|u|\Big ).
	\end{align}
    Hence, in view of Lemma~\ref{lem:const_nl}, by selecting $\mu\in (0, \nu]$ we have that
	$\sup\dot{\overline{W}}_F(Cx) \leq \sup\limits_{\substack{u \in - {\boldsymbol{\Psi}}(Cx)}} (-\omega|u|)\leq	- c \omega,$
	where $\omega=\lambda_1\big (c -2\frac{\lambda_2 \mu}{\lambda_1}\big)>0$, thus concluding the proof.
\end{proof}

We are now ready to prove Theorem~\ref{thm:OFTS}. To prove the OFTS property of system \eqref{sys_reg_eq}, we proceed by steps. We first show that, for solutions to \eqref{sys_reg_eq} initialized in a neighborhood of the origin, the corresponding output converges to the origin in finite-time and then, leveraging the GAS property of the origin for \eqref{sys_reg_eq}, we prove OFTS of \eqref{sys_reg_eq}. 

\begin{proofname}{Proof of Theorem~\ref{thm:OFTS}.}
We start by proving that solutions initialized sufficiently close to the origin converge to $\ker(C)$ in finite time by integrating \eqref{eq:dot_aux_V}. To do so, we recall that, by the GAS property of the origin, there exists $\kappa>0$ such that solutions starting in $\kappa \mathbb{B}$ will not leave $\mu \mathbb{B}$, with $\mu$ as in Proposition~\ref{prop:aux_V_prop} and we note that the set $\mu\mathbb{B}_n \cap \ker (C)$ is forward invariant for any solution starting $\kappa\mathbb{B}_n \cap \ker (C)$. Indeed,  suppose that there exists a solution $x_{\text{bad}}$ to \eqref{sys_reg_eq} such that $x_{\text{bad}}(0)\in \kappa\mathbb{B}_n\cap\ker (C)  $ and  $x_{\text{bad}}(t^*)\notin \mu\mathbb{B}_n \cap \ker (C)  $ for some $t^*>0$ with $t^* \in \dom x_{\text{bad}}$. Since $x_{\text{bad}}$ is continuous with respect to the time, we can choose ${t}^*>0$ such that $ x_{\text{bad}} (t) \in \kappa\mathbb{B}_n \cap \ker (C)$ for all $t \in  [0,{t}^*)$ and $x_{\text{bad}} ({t}^*) \in \mu\mathbb{B}_n \setminus  \ker(C)$. Hence, from  \eqref{eq:aux_V} and \eqref{eq:dot_aux_V},  and from the fact that $W$ is positive definite on $\mu\mathbb{B}_n$ and non-pathological,  we have  $0=W(Cx_{\text{bad}} (t))<W(Cx_{\text{bad}} ({t}^*))$, for all $t \in  [0,{t}^*)$, which establishes a contradiction by the continuity property of $W$. Consequently, solutions cannot leave $ \mu\mathbb{B}_n \cap  \ker (C)$ after reaching the set $\kappa\mathbb{B}_n \cap \ker (C)$. Therefore, by combining \eqref{eq:dot_aux_V} with the fact that $W$ is non-pathological, and the forward invariance of $\mu \mathbb{B} \cap \ker (C)$ for solutions starting in $\kappa \mathbb{B} \cap \ker (C)$,  for any solution $x$ initialized so that $x(0)\in \kappa\mathbb{B}_n \setminus  \ker(C)$, we obtain by integration for any $t\in \dom x$ such that  $x(t) \in \mu\mathbb{B}_n \setminus  \ker(C)$
\begin{align}
	\label{eq:V_fin_time}
	\begin{split}
		W(Cx(t))\leq -c \omega t + W(Cx(0)),
	\end{split}
\end{align}
and thus 
\begin{align}
	    \nonumber
		W(Cx(t))\leq\max(-c \omega t &+ W(Cx(0)) ,0), \\ \label{eq:V_fin_time_bound} & \forall x(0)\in \kappa\mathbb{B}_n, \, \forall t \in \real_{\geq 0}.
\end{align}
Thus, in view of \eqref{eq:V_fin_time_bound} and by the GAS property of the origin, we conclude that, for any solutions starting in $\kappa \mathbb{B}$, there exists a $T_y$, depending on $\kappa$, such that $x(t) \in \kappa \mathbb{B} \cap \ker(C)$ for any $t\geq T_y$.  We now leverage the GAS property of the origin to prove that \eqref{sys_reg_eq} is OFTS. We recall that, for any solution $x$ to \eqref{sys_reg_eq}, by the GAS property of the origin there exists a time $T_\kappa \geq 0$ such that $x(t) \in \kappa \mathbb{B}$ for all $t \geq T_\kappa$. Therefore, we conclude that $y(t)=\boldsymbol 0_p$ for all $t \geq T:= T_\kappa + T_y$. We have proved that, for any solution $x$, there exists $T \geq 0$ such that $y(t)=\boldsymbol 0_p$, for all $t\geq T$. Moreover, system \eqref{sys_reg_eq} is oGAS because it is GAS from item (i) of Assumption~\ref{ass:ofts} and because $|y| \leq |C|  |x|$. Therefore, system  \eqref{sys_reg_eq} is OFTS.

Finally, we prove that system \eqref{sys_reg_eq} is also SIoLAS. Indeed, combining \eqref{eq:aux_V_sandwitch} and  \eqref{eq:V_fin_time_bound} yields, for any solution $x$ with $x(0) \in \kappa \mathbb{B}$,
\begin{align}
\nonumber
	|y(t)|&\leq \alpha_4^{-1}(\max(-c \omega t + \alpha_5(Cx(0)) ,0))\\   &=:\beta_\circ(|y(0)|,t), \quad \forall x(0)\in  \kappa\mathbb{B}_n, \quad \forall t \in \real_{\geq 0},
	\label{eq:KL_out_3}
\end{align}
with $\beta_\circ \in \KL$, thus ending the proof.
\end{proofname}
\section{Applications}
\label{sec:lure_app}
In this section, we present two applications of the results of Sections~\ref{sec:asy_stability} and \ref{sec:ftp}.
\subsection{Mechanical system affected by friction \cite{Eindhoven17}}
\label{subsec:friction}
Consider the rotor dynamic system with friction system given in \cite[Sec. 5]{Eindhoven17}, i.e.,
\begin{align}
  \!\!\begin{bmatrix}	
    \dot \alpha\\
    \dot \omega_u\\
	\dot \omega_\ell
 \end{bmatrix} 
 \!\!\in \!\! \label{eq:mech_dyn_og}
  \begin{bmatrix}	
	\omega_u - \omega_\ell \\
    -\frac{k_\theta}{J_u} \alpha  -\frac{b}{J_u}(\omega_u-\omega_\ell)-\frac{1}{J_u} T_{fu}(\omega_u) + \frac{k_u}{J_u} v \\
    \frac{k_\theta}{J_\ell} \alpha   + \frac{b}{J_\ell}(\omega_u-\omega_\ell)-\frac{1}{J_\ell} T_{f\ell} (\omega_\ell) 
 \end{bmatrix}, 
\end{align}
with $x=(\alpha,\omega_u,\omega_\ell) \in \real^3$, where $\alpha$ is the angular mismatch between two rotating discs connected by an angular spring and an angular dumper, and $\omega_u$ and $\omega_\ell$ are the angular velocities of these two discs. Scalars $J_u$, $J_\ell$, $k_u$, $k_\theta$ and $b$ are positive system parameters whose values are reported in Table~\ref{tab:params}.  The control input $v \in \real$ is used for state-feedback stabilization, while the set-valued maps $T_{fu}$ and $T_{f\ell}$ in \eqref{eq:mech_dyn_og} are defined as
\begin{align*}
    &T_{fu}(s) := \,\,
    \begin{cases}
	\, f_u(s)  \sign(s), \hspace{17.5 mm} \forall s \in \real \setminus \{0\} \\
	\, [-f_{u,\circ}+\Delta f_u,f_{u,\circ} +\Delta f_u], \quad  \text{otherwise}, 	
	\end{cases}  \\
	& f_u(s):= f_{u,\circ}+ \Delta f_u \sign(s) + q_1 |s|  +  q_2 s, \quad  \forall s \in \real,\\
    &T_{f\ell}(s) := \,\,
    \begin{cases}
	\, f_\ell(s)  \sign(s), \hspace{8 mm} \forall s \in \real \setminus \{0\} \\
	\, [-f_{\ell,\circ},f_{\ell,\circ}], \hspace{10.5 mm}  \text{otherwise}, 	
	\end{cases}  \\
	& f_\ell(s):= f_{\ell,\circ}+ (\Delta f_\ell - f_{\ell,\circ}) e^{-q_3|s|} + q_4 |s|, \quad  \forall s \in \real,
\end{align*}
 for suitable positive scalars $f_{u,\circ}$, $f_{\ell,\circ}$, $\Delta f_u$, $\Delta f_\ell$, $q_1$, $q_2$, $q_3$ and $q_4$ we give in Table~\ref{tab:params} and with function $\sign:\real \rightarrow [-1,1]$ defined as  $\sign(s)=1$ if  $s>0$, $\sign(s)=-1$ if  $s<0$, and $\sign(s)=0$ if $s=0$, and for which we have that \eqref{eq:sec_cond} is satisfied with $\zeta_1=\zeta_2=\infty$.

Like in \cite[Ch. 6]{doris2007output}, by considering the selection $v=v_p+v_{\text{lin}}$ in \eqref{eq:mech_dyn_og}, where $v_{p_1}:=K x$, $K=[k_1,k_2,k_3]\in \real^{3 \times 1}$ and $v_{\text{lin}}:=\frac{1}{k_u} T_{fu}(\omega_u)$, we obtain 
\begin{align}
\label{eq:mech_dyn_lin}
 \!\!\!\!\begin{bmatrix}	
    \alpha\\
     \omega_u\\
	 \omega_\ell
 \end{bmatrix} \!\!\!
  \in  \!\!
  \begin{bmatrix}	
	\!\omega_u - \omega_\ell \!\\
    \!-\frac{k_\theta}{J_u} \alpha \! -\! \frac{b}{J_u}(\omega_u\!-\omega_\ell)\! +\! \frac{k_u}{J_u} (k_1  \alpha + k_2 \omega_u + k_3 \omega_\ell)  \!\\
    \!\frac{k_\theta}{J_\ell} \alpha   + \frac{b}{J_\ell}(\omega_u-\omega_\ell)-\frac{1}{J_\ell} T_{f\ell} (\omega_l)\!
 \end{bmatrix}\!, 
\end{align}
which can be written in the Lur'e form \eqref{sys_reg_eq}, with $n=3$ and $p=1$, and $A=A_{\text{free}} + H_1 K +  H_2$,
\begin{equation*}
	A_{\text{free}}= \begin{bmatrix}
	0 &
	1 &
	-1\\
	-\frac{k_\theta}{J_u}  &
    -\frac{b}{J_u}&
	\frac{b}{J_u}\\
	\frac{k_\theta}{J_\ell}&
	\frac{b}{J_\ell}&
	-\frac{b}{J_\ell}
	\end{bmatrix}, \,
	H_1 K= \begin{bmatrix}
	0  &
    0  &
	0 \\
	\frac{k_u k_1}{J_u}&
    \frac{k_u k_2}{J_u}&
    \frac{k_u k_3}{J_u}\\
	0&
	0&
    0
	\end{bmatrix},   
\end{equation*}
$H_2=\diag(0, 0, \frac{m}{J_u})$ and $m\in \real$, $B=(0, 0, \frac{1}{{J_\ell}})$, $C=[0, 0, 1]$ and $\boldsymbol \Psi(C x)=\boldsymbol \Psi(\omega_\ell)=T_{f\ell} (\omega_\ell)+m \omega_\ell$.  

\begin{table}
\centering
\label{tab:params}
    \begin{tabular}{lll}
    \toprule
    $b$ &$[N m^2/\text{rad } s]$& 0  \\
    $f_{u,\circ}$  &$[\text{N } \text{m}]$ & 0.38 \tabularnewline 
    $\Delta f_u$ &$[\text{N }  \text{m}]$ & -0.006  \tabularnewline
    $f_{\ell,\circ}$ &$[\text{N }  \text{m}]$ & 0.0009  \tabularnewline
    $\Delta f_\ell$ &$[\text{N }  \text{m}]$ & 0.68  \tabularnewline
    $J_u$  &$[\text{kg } \text{m}^2]$ & 0.4765  \\ 
    $J_\ell$ &$[\text{kg } \text{m}^2]$ & 0.035 \\ 
    $k_u$ &$[\text{N }  \text{m}/\text{V}]$ & 4.3228   \tabularnewline
    $k_\theta$ &$[\text{N } \text{m}/\text{rad}]$ & 0.075  \tabularnewline 
    $q_1$ &$[\text{kg } \text{m}^2/\text{rad s}]$& 2.4245  \\ 
    $q_2$ &$[\text{kg } \text{m}^2/\text{rad s}]$& -0.0084  \\ 
    $q_3$ &$[s/\text{rad}]$& 0.05  \\ 
    $q_4$ &$[\text{kg } \text{m}^2/\text{rad s}]$& 0.26  \\
    \bottomrule
    \end{tabular}
    \caption{Parameters identifying the system given in \cite[Sec. 5]{Eindhoven17}.}
\end{table}

Assumption \ref{ass:strict_pass} is satisfied with the selection $m=0.052$, $\Gamma=\gamma_1=10$, $\eta=8.492$, $K=[-12.8282,3.7216,-8.4816]$ and
\begin{equation*}
	P= \begin{bmatrix}
	0.5636 &
	0.0340 &
	0.3793\\
	0.0340  &
    0.0062&
	0.0186\\
	0.3793&
	0.0186&
	0.2642
	\end{bmatrix}. \,   
\end{equation*}
Since Assumption~\ref{ass:dec_nl} is also satisfied, Theorem \ref{thm:KLbound} implies that the origin is GAS for \eqref{sys_reg_eq}, thus retrieving the result originally presented in \cite{Eindhoven17}. In addition, because Assumption \ref{ass:ofts} holds for the considered system, we establish here, from Theorem \ref{thm:OFTS}, that system \eqref{sys_reg_eq} is OFTS and SIoLAS, which is a novelty compared to \cite{Eindhoven17}.

\subsection{Cellular neural networks from \cite{forti2003global}}
\label{sec:neuralnet}
In \cite{forti2003global}, cellular neural networks are modeled  by system \eqref{sys_non_reg}  (see \cite[eq. (N1)-(N2)]{forti2003global}), where the system data satisfies the next property according to \cite[Prop. 3 and 4]{forti2003global}.
 \begin{prop}
 \label{ass:celneuralnet}
The following holds for system \eqref{sys_non_reg}.
 \begin{enumerate}[label=(\roman*)]
     \item  $A$ is a diagonal, negative definite matrix.
     \item  $B$ is LDS (as per Assumption~\ref{ass:ofts}).
     \item  $C=I_n$.
     \item  For any $i\in\{1,\dots,n\}$, function $\psi_i$ is nondecreasing, i.e., for any $a>b \in \dom \psi_i$ it holds that $\psi_i(a) \geq \psi_i(b)$, is piecewise continuous and satisfies Assumption~\ref{ass:dec_nl} with $\zeta_i=+\infty$ and item~(iii) of Assumption~\ref{ass:ofts}. \hfill $\square$
 \end{enumerate}    
 \end{prop}
Property~\ref{ass:celneuralnet} trivially implies Assumption~\ref{ass:dec_nl} and  items (i) and (iii) of Assumption~\ref{ass:ofts}. We show below that it also implies item (ii) of Assumption~\ref{ass:ofts} so that we can invoke Theorems~\ref{thm:KLbound} and \ref{thm:OFTS} to prove GAS of the origin for system \eqref{sys_reg_eq} and that system \eqref{sys_reg_eq} is SFTS, thus providing alternative proofs of the stability results given in  \cite[Thm. 3 and 4]{forti2003global}. Indeed, we recall that, by proving stability properties for system \eqref{sys_reg_eq}, we ensure the same stability properties for the Krasovskii solutions of \eqref{sys_non_reg}.
\begin{lemma}
\label{lem:cnn_props}
Suppose that system \eqref{sys_non_reg} satisfies Property~\ref{ass:celneuralnet}. Then the origin is GAS for system \eqref{sys_reg_eq}, and system \eqref{sys_reg_eq} is SFTS.
\end{lemma}
\begin{proof}
We prove below that there exist matrices $\Gamma>0$ diagonal, $P=P^\top >0$ and a scalar $\eta>0$ satisfying \eqref{eq:diss_sys1n}. Since $B$ is LDS, there exists a $\Gamma>0$ diagonal such that $\Gamma B + (\Gamma B)^\top=:\Sigma>0$ and such that $\Gamma A \leq -I_n$. With this selection, we can rewrite matrix $\overline M$ in \eqref{eq:diss_sys1n} as, 
\begin{equation}
\label{eq:spr_mat}	
	\overline M=\begin{bmatrix}
	P A + A^\top P&
	PB \\
	B^\top P&
		-\Sigma 
	\end{bmatrix} + \diag(\eta I_n,\boldsymbol 0_n),
\end{equation}
noting that $Z$ is the null matrix due to item (iv) of Property~\ref{ass:celneuralnet}. Define \begin{align*}
\widetilde M &:=\begin{bmatrix}
	P A + A^\top P&
	PB \\
    B^\top P&
		-\Sigma 
	\end{bmatrix}\\
 &=\begin{bmatrix}
	P&
	0 \\
	0&
		I_n
	\end{bmatrix}
 \underbrace{\begin{bmatrix}
	S A^\top  + A S&
	B \\
	B^\top&
		-\Sigma 
	\end{bmatrix}}_{\hspace{2mm}=:N}
 \begin{bmatrix}
	P&
	0 \\
	0&
		I_n
	\end{bmatrix},
\end{align*} 
where $S=P^{-1}$. Since $A$ is Hurwitz by item (i) of Property~\ref{ass:celneuralnet}, there exists $S_\circ=S_\circ^\top>0$ such that $S_\circ A^\top  + A S_\circ= \Pi<0$. Therefore, by selecting $S=\alpha S_\circ$ with $\alpha>0$, to be chosen, we have that 
\begin{equation}
     N=\begin{bmatrix}
	\alpha \Pi&
	B \\
	B^\top&
		-\Sigma 
	\end{bmatrix} < 0 \quad \forall \alpha>\alpha^\star,
\end{equation}
where $\alpha^\star>0$ satisfies $-\alpha^\star \lambda_{\Pi}>|B \Sigma B^\top|$, with $\lambda_{\Pi}>0$ denoting the smallest eigenvalue of $\Pi$. Hence, with the given selection of $\alpha$ and $P$, matrix $N$ and thus $\widetilde M$ are negative definite. Therefore, by selecting $0<\eta<-|\widetilde M|$ we have that $\overline M \leq 0$ thus proving \eqref{eq:diss_sys1n} and item (ii) of Property~\ref{prop:alt_prop}. Consider now item (iii) of Property~\ref{prop:alt_prop} and note that matrix $H=\Gamma A < 0$ is diagonal negative definite and satisfies $\Gamma C A=\Gamma A = H = H C$. Since $Z=O_p$ due to item (iv) of Property \ref{ass:celneuralnet}, then item (iii) of Property~\ref{prop:alt_prop} holds and we can invoke Lemma~\ref{lem:alt_gas} to certify that the origin is GAS for system \eqref{sys_reg_eq} and render Property~\ref{ass:celneuralnet}. Furthermore, since Assumptions \ref{ass:dec_nl} and \ref{ass:ofts} hold, then, by Theorem \ref{thm:FTS}, system \eqref{sys_reg_eq} is also SFTS  because $C$ is invertible. 
\end{proof}
We envision applying our results to a broader class of neural networks with piecewise continuous activation functions. Due to the short length of a technical note submission, we do not pursue such generalizations here and we regard them as future work.
\section{Conclusion}
We have analyzed the stability of the origin for Lur'e systems with piecewise continuous nonlinearities. We have first established the global asymptotic stability of the origin under a milder sector condition compared to \cite{Eindhoven17} and by relying on a different, algebraic Lyapunov proof based on the concept of set-valued Lie derivative. We have then presented conditions under which finite-time stability properties can or cannot be established for the considered class of systems. These results have been applied to two engineering systems of interest: mechanical systems with friction and cellular neural networks.

Future research directions may include: systems affected by exogenous disturbances; weak stability analysis for the considered class of systems in the sense that only some solutions exhibit the desired stability properties; as well as the synchronization of interconnected Lur'e systems with piecewise continuous nonlinearities following the path of paved by \cite{tang2017finite,brogliato2009observer}. 
\bibliographystyle{plain}
\bibliography{refs}
\end{document}